\documentclass[12pt]{amsart}

\usepackage{a4wide,amsfonts,graphicx,amssymb,stmaryrd,amscd,amsmath,latexsym,amsbsy,xypic,amsthm,amsbsy,color,multicol}
\xyoption{all}
\usepackage{marginnote,mathtools,epsfig,enumitem,dsfont,stmaryrd,array,calc,rotating,bm,bbm,mathrsfs} 
 
\usepackage{tikz}
\usetikzlibrary{decorations.markings}
\usetikzlibrary{calc} 
\usetikzlibrary{arrows} 
\usetikzlibrary{patterns} 
\usetikzlibrary{decorations.pathreplacing} 

\numberwithin{equation}{section}
\theoremstyle{plain}
 
\newtheorem{thm}{Theorem}[section]

\newtheorem{defi}[thm]{Definition}
\newtheorem{lem}[thm]{Lemma}

\theoremstyle{remark}
\newtheorem{rema}[thm]{Remark}

\newcommand{\Z}{\mathbb{Z}}

\newcommand{\C}{\mathbb{C}}
\newcommand{\V}{\mathbb{V}}

\newcommand{\ket}[1]{\left|#1\right\rangle}      
\newcommand{\bra}[1]{\left\langle #1\right|}     

\newcommand{\PROD}[3]{\mathop{\overrightarrow\prod}\limits_{#1 \le #2 \le #3 }}
\newcommand{\iPROD}[3]{\mathop{\overleftarrow\prod}\limits_{#1 \le #2 \le #3 }}

\newcommand{\hypref}[2]{\ifx\href\asklfhas #2\else\href{#1}{#2}\fi}
\newcommand{\Secref}[1]{Section~\ref{#1}}

\renewcommand{\eqref}[1]{(\ref{#1})}

\def\[{\begin{equation}}
\def\]{\end{equation}}
\def\<{\begin{eqnarray}}
\def\>{\end{eqnarray}}

\title[]{Continuous representations of scalar products \\ of  Bethe vectors}
\author{W. Galleas}
\address{II. Institut f\"ur Theoretische Physik, Universit\"at Hamburg, Luruper Chaussee 149, 22761 Hamburg, Germany.}
\email{wellington.galleas@desy.de}
\thanks{The work of W.G. is supported by the German Science Foundation (DFG) under the Collaborative Research Center (SFB) 676: Particles, Strings and the Early Universe.}

\subjclass[2000]{}
\begin{document}

\vspace{2em}
\begin{abstract}
We present families of single determinantal representations of \emph{on-shell} scalar products of Bethe vectors. Our families of 
representations are parameterized by a continuous complex variable which can be fixed at convenience. Here we consider Bethe vectors
in two versions of the six-vertex model: the case with boundary twists and the case with open boundaries.
\end{abstract}

\maketitle
\vspace{-8cm}
\begin{flushright}
{\footnotesize ZMP-HH/16-16}
\end{flushright}
\vspace{10cm}

\setcounter{tocdepth}{2}

\section{Introduction} \label{sec:INTRO}

Integrability in the sense of Liouville is essentially based on the possibility of solving exactly equations of motion arising
in  classical Hamiltonian systems. However, the advent of \emph{quantum mechanics} has changed the scenario and physical quantities 
associated with several important systems are not \emph{a priori} described by differential and/or difference equations. For instance, this is the case in 
the isotropic Heisenberg spin-chain, or \textsc{xxx} model for short \cite{Heisenberg_1928}.

\subsection{Bethe vectors} 
The one-dimensional \textsc{xxx} model with nearest-neighbor interactions and periodic boundary conditions was solved by
Bethe in the seminal work \cite{Bethe_1931}. Bethe's solution involves an \emph{ansatz} for the model's eigenvectors containing
free parameters; which are then fixed by the hamiltonian eigenvalue problem. The method put forward by Bethe was found suitable for a
large number of spin-chain hamiltonians and it is nowadays widely known as \emph{Bethe ansatz}. The constraints fixing the ansatz parameters
are in their turn referred to as \emph{Bethe ansatz equations}. 

Interestingly, the technique pioneered by Bethe is not limited to spin-chain hamiltonians and it has been used for solving a variety of
problems. For instance, in \cite{Lieb_1967} Lieb showed that Bethe's method can also be used for diagonalizing the transfer matrix of cases
of the six-vertex model. Lieb's analysis was further generalized in \cite{Lieb_1967a, Lieb_1967b, Lieb_1967c, Sutherland_1967}; and in all cases
the transfer matrices eigenvectors were found to coincide with the ones diagonalizing the anisotropic Heisenberg chain hamiltonian \cite{Yang_Yang_1966}.
The latter is also known as \textsc{xxz} model and it was proposed by Bloch as model for remanent magnetization \cite{Bloch_1930, Bloch_1932}.
At first it was not clear why eigenvectors of the six-vertex and \textsc{xxz} models would coincide and this was only clarified by McCoy and
Wu in \cite{McCoy_Wu_1968}. More precisely, McCoy and Wu showed that the transfer matrix of the six-vertex model formulated in 
\cite{Sutherland_1967} commutes with the hamiltonian of the \textsc{xxz} spin-chain.
This property was further developed by Baxter culminating in the concept of \emph{commuting transfer matrices} \cite{Baxter_1971f, Baxter_1971}. 
In particular, the latter was shown to be achieved by local relations nowadays known as \emph{Yang-Baxter equation}.

The existence of commuting transfer matrices has also led to further insights into the structure of models exactly solvable by Bethe ansatz. For instance,
the relation between Bethe ansatz equations and commuting transfer matrices was then made manifest in \cite{Baxter_1971} through the so called \textsc{t-q}
relations. Subsequently, the algebraic structure underlying the construction of Bethe vectors was unveiled by the 
\emph{Quantum Inverse Scattering Method}  proposed in \cite{Sk_Faddeev_1979, Takh_Faddeev_1979}. Such construction of Bethe vectors
is also known as \emph{algebraic Bethe ansatz} and this is the formulation we shall consider in the present paper.

\subsection{Scalar products} 

Within the framework of the algebraic Bethe ansatz one finds that eigenvectors of \emph{integrable} spin-chains can be built through the 
action of \emph{creation} operators on a \emph{pseudo-vacuum} vector. This construction is analogous to the solution of the quantum mechanical 
harmonic oscillator through the Heisenberg-Weyl algebra or the hydrogen atom through the $\mathfrak{so}(4)$ Lie algebra \cite{Pauli_1926, Fock_1935}. 
One advantage of this approach, when compared to Bethe's original methodology \cite{Bethe_1931}, is that it offers compact representations
for the model's eigenvectors. In particular, as far as the study of correlation functions are concerned, such algebraic representations of Bethe 
vectors have proved quite suitable. See for instance \cite{Korepin_book} and references therein. 

The computation of correlation functions does not necessarily require the model's eigenvectors. In fact, one only needs to evaluate 
certain expected values and the first quantity of interest is the scalar product of Bethe vectors \cite{Korepin_1982}. 
It is also customary in the literature  to distinguish Bethe vectors based on the conditions imposed on their Bethe ansatz parameters. 
Hence, we refer to the case where such parameters are free as \emph{off-shell Bethe vector}; whilst the case satisfying Bethe ansatz equations
is referred to as \emph{on-shell Bethe vector}. In this way one can compute scalar products between two off-shell vectors, two on-shell vectors; 
and also between one off-shell and one on-shell vector. Here we refer to those cases respectively as off-shell scalar products, norms and on-shell scalar products.

As far as the six-vertex model with diagonal boundary twists is concerned, off-shell scalar products have been studied in 
\cite{Korepin_1982, Galleas_SCP, deGier_Galleas_2011}. In particular, in the works \cite{Galleas_SCP, deGier_Galleas_2011} we have obtained compact 
expressions for such scalar products in terms of multiple contour integrals. Norms were computed in \cite{Korepin_1982, Korepin_book} proving Gaudin's 
hypothesis that such quantities were given by Jacobian determinants \cite{Gaudin_1972, Gaudin_McCoy_Wu_1981}. Interestingly, on-shell scalar products
can also be written as determinants and such formula was obtained by Slavnov in \cite{Slavnov_1989}.

\subsection{Algebraic-functional approach} 
In the course of studying scalar products of Bethe vectors, Korepin has also introduced the so called six-vertex model with domain-wall boundary
conditions \cite{Korepin_1982}. The latter has found numerous applications ranging from enumerative combinatorics \cite{Kuperberg_1996} 
to the study of gauge theories \cite{Szabo_2012}. In particular, a set of conditions characterizing the partition function of the six-vertex 
model with domain-wall boundaries is among the results of \cite{Korepin_1982}. Alternatively, the same partition function can be studied through certain
functional equations originating from the Yang-Baxter algebra \cite{Galleas_2010, Galleas_proc, Galleas_2013, Galleas_2016a, Galleas_2016b}.
This is a result of the \emph{algebraic-functional method} originally proposed for spectral problems in \cite{Galleas_2008} and subsequently 
extended to such partition functions in \cite{Galleas_2010}. 

Scalar products of Bethe vectors can also be tackled through the algebraic-functional method. This was demonstrated in \cite{Galleas_SCP}
for the six-vertex model with boundary twists, and in \cite{Galleas_openSCP} for the case of open boundaries. In particular, in the works
\cite{Galleas_SCP, Galleas_openSCP} we have obtained a system of two functional equations determining such off-shell scalar products. 
Equations governing on-shell scalar products are also available and they have been discussed in \cite{Galleas_SCP, Galleas_Lamers_2015}.
One important feature of the equations describing the on-shell case is that they exhibit the same structure as the functional equation derived in
\cite{Galleas_2013, Galleas_proc} for the partition function of the six-vertex model with domain-wall boundaries. The importance of this result 
is two-fold: on the one hand it shows that several quantities of interest can be accommodated in such type of equation.
On the other hand, it paves the way for using the recently formulated method of \cite{Galleas_2016a, Galleas_2016b} and consequently
obtaining families of single determinantal representations for on-shell scalar products.

\subsection{Goals of this paper} 
The approach put forward in \cite{Galleas_2016a, Galleas_2016b} for solving functional equations originating from the algebraic-functional
method depends mainly on the structure of the equation rather than on the particular form of its coefficients. In this way, one might expect it
can be employed for solving the equations governing on-shell scalar products presented in \cite{Galleas_SCP, Galleas_Lamers_2015}. Demonstrating 
the feasibility of our method for such equations is the main goal of the present paper. In particular, we intend to show that our approach yields more 
than one single determinantal representation for on-shell scalar products. As a matter of fact, it produces several families of continuous
determinantal representations. Here we refer to continuous determinantal representations as the determinant of a matrix depending 
non-trivially on a continuous complex variable, such that the computation of the determinant eliminates this dependence. Hence, this variable can be regarded
as a parameterization of our continuous families of representations since different choices of this parameter modifies the matrix we are considering but not 
the final scalar product. Moreover, in the present paper we shall consider on-shell scalar products associated with two types of six-vertex models:
the case having diagonal boundary twists and the case with diagonal open boundaries. Although both cases have already been considered in
the literature \cite{Slavnov_1989, Wang_2002, Kitanine_2007}, here we shall obtain a large number of continuous families of representations for
such quantities in a single step. In particular, our representations will be given by single determinants allowing for a straightforward derivation of the 
partial homogeneous limit.

\subsection{Outline} 
The present paper is organized as follows. In \Secref{sec:TWIST} we consider the six-vertex model with boundary twists. We restrict ourselves
to presenting only the definitions required for the computation of on-shell scalar products through the algebraic-functional method. 
\Secref{sec:OPEN} is devoted to the six-vertex model with open boundary conditions. Similarly to \Secref{sec:TWIST}, we largely rely 
on results previously presented in the literature in \Secref{sec:OPEN}. Concluding remarks are then presented in \Secref{sec:CONCL}.

\section{Six-vertex model: boundary twists} \label{sec:TWIST}

Vertex models consist of a collection of graphs embedded in a lattice, statistical weights for graphs configurations and boundary conditions.
Such formulation is in fact a generalization of the \emph{ice model} proposed by Pauling in 1935 in order to describe the ice residual 
entropy \cite{Pauling_1935}. Two-dimensional vertex models are the most notorious ones to date; and for such models the concept of \emph{integrability}
in the sense of Baxter is well established \cite{Baxter_book}. However, it is important to emphasize here that \emph{integrability} in this context does 
not refer to the possibility of solving exactly differential equations. Vertex models are systems of classical Statistical Mechanics, and in that case
one needs to evaluate the summation defining the model's partition function and other physical quantities in a closed form.
In the present paper we shall restrict our attention to a particular integrable vertex model, alias trigonometric six-vertex model. Several cases of the
trigonometric six-vertex model was firstly studied by Lieb \cite{Lieb_1967, Lieb_1967a, Lieb_1967b, Lieb_1967c} and its full version was put forward by
Sutherland in \cite{Sutherland_1967}. Definitions required for the construction of Bethe vectors in the six-vertex model with diagonal boundary twists
will be introduced in what follows.

\subsection{Definitions and conventions} \label{sec:DEF_twist}
Here we shall also refer to the six-vertex model with boundary twists as toroidal six-vertex model. The literature devoted to this particular vertex model
is quite extensive and we restrict ourselves to describing briefly the mathematical objects required for the definition of Bethe vectors. 
Physical assumptions underlying the six-vertex model can be found in \cite{Lieb_1967, Baxter_book} and references therein. 
One interesting aspect of the toroidal six-vertex model is that it allows for an operatorial description leading to an algebraic formulation of Bethe 
vectors. Also, we stress our presentation will heavily rely on results previously obtained in \cite{Galleas_SCP, Galleas_openSCP, Galleas_Lamers_2015}.

\subsection{$\mathcal{U}_q [\widehat{\mathfrak{gl}_2}]$-intertwinner} The statistical weights of the integrable six-vertex model can be encoded in a 
matrix $\mathcal{R}$ intertwining tensor product representations of the $\mathcal{U}_q [\widehat{\mathfrak{gl}_2}]$ algebra. As for the six-vertex model
we consider the vector space $\mathbb{V} = \C v_0 \oplus \C v_{\bar{0}}$ with basis vectors $v_0, v_{\bar{0}} \in \C^2$. The intertwinner of interest
$\mathcal{R} \colon \C \to \text{End} ( \mathbb{V} \otimes \mathbb{V} )$ then reads
\[ \label{rmat}
\mathcal{R} (x) = \begin{pmatrix} a(x) & 0 & 0 & 0 \\ 0 & b(x) & c(x)  & 0 \\ 0 & c(x) & b(x) & 0 \\ 0 & 0 & 0 & a(x) \end{pmatrix} \; ,
\]
with non-null entries defined as $a(x) \coloneqq \sinh{(x + \gamma)}$, $b(x) \coloneqq \sinh{(x)}$ and $c(x) \coloneqq \sinh{(\gamma)}$. They are functions on complex parameters
$x$ and $\gamma$, respectively referred to as \emph{spectral} and \emph{anisotropy} parameters. Throughout this work $\gamma \in \C$ is fixed
and we omit its dependence on the LHS of \eqref{rmat}.

\subsection{Yang-Baxter algebra}
Let $\mathscr{A} (\mathcal{R})$ denote the algebra generated by the non-commutative matrix entries of $\mathcal{L} \in \text{End}(\V)$ satisfying
the relation
\[ \label{yba}
\mathcal{R}_{12} (x_1 - x_2) \; \mathcal{L}_1 (x_1) \mathcal{L}_2 (x_2) = \mathcal{L}_2 (x_2) \mathcal{L}_1 (x_1) \; \mathcal{R}_{12} (x_1 - x_2) 
\]
in $\text{End} (\V \otimes \V \otimes \V_{\mathcal{Q}})$. Relation \eqref{yba} considers the standard tensor leg notation. The algebra $\mathscr{A} (\mathcal{R})$
is usually referred to as Yang-Baxter algebra and its representations consist of a module $\V_{\mathcal{Q}}$ together with a meromorphic function
$\mathcal{L} (x)$ on $\C$ with values in $\text{End} (\V \otimes \V_{\mathcal{Q}})$. As for the six-vertex model, we have $\V \simeq \C^2$ and we write
\[ \label{abcd}
\mathcal{L} (x) \eqqcolon \begin{pmatrix} \mathcal{A}(x) & \mathcal{B}(x) \\ \mathcal{C}(x) & \mathcal{D}(x) \end{pmatrix} \; .
\]
In this way $\mathcal{A}(x)$, $\mathcal{B}(x)$, $\mathcal{C}(x)$ and $\mathcal{D}(x)$ are operators on $\C$ with values in $\text{End} (\V_{\mathcal{Q}})$.

\subsection{Modules over $\mathscr{A} (\mathcal{R})$} Representations of the Yang-Baxter algebra $\mathscr{A} (\mathcal{R})$ 
consist of pairs $(\V_{\mathcal{Q}} , \mathcal{L})$. The case formed by $\V_{\mathcal{Q}} = \V$ and $\mathcal{L} (x) = \mathcal{R}(x - \mu)$ is usually
called fundamental representation with evaluation point $\mu$. Next, we let $L \in \Z_{> 0}$ and look for $\mathscr{A} (\mathcal{R})$-modules 
with $\V_{\mathcal{Q}} = \V^{\otimes L}$. The latter can be easily obtained by noticing the following: if $(\V_{\mathcal{Q}}' , \mathcal{L}')$ and
$(\V_{\mathcal{Q}}'' , \mathcal{L}'')$ are $\mathscr{A} (\mathcal{R})$-modules, then so is $(\V_{\mathcal{Q}} , \mathcal{L})$ with 
$\V_{\mathcal{Q}} = \V_{\mathcal{Q}}' \otimes \V_{\mathcal{Q}}''$ and $\mathcal{L} (x) = \mathcal{L}' (x - x') \mathcal{L}'' (x - x'')$. Now write $\V_0 \simeq \V$ and let $\mathcal{T}_0 \in \text{End}(\V_0 \otimes \V^{\otimes L})$
be referred to as \emph{monodromy matrix}. Hence, one finds that the pair $(\V^{\otimes L} , \mathcal{T}_0 )$ is an $\mathscr{A} (\mathcal{R})$-module
with 
\[ \label{mono}
\mathcal{T}_0 (x) \coloneqq \PROD{1}{j}{L} \mathcal{R}_{0 j} (x - \mu_j) \; .
\]
The parameters $\mu_j \in \C$ in \eqref{mono} denote \emph{lattice inhomogeneities} and in what follows we shall consider the $\mathscr{A} (\mathcal{R})$-module
$(\V^{\otimes L} , \mathcal{T}_0 )$ built out of \eqref{mono} and \eqref{rmat}. Representations of the operators $\mathcal{A}$, $\mathcal{B}$, $\mathcal{C}$
and $\mathcal{D}$ are then obtained through the identification of \eqref{abcd} and \eqref{mono}.

\begin{rema}
Throughout this paper we fix the inhomogeneity parameters $\mu_j \in \C$ and omit them from the arguments of functions.
\end{rema}

\subsection{Bethe vectors and scalar products} 
Let $\mathfrak{h}$ denote the Cartan subalgebra of $\mathfrak{gl}_2$ and let $\V^{\otimes L}$ be a diagonalizable $\mathfrak{h}$-module. Then we say
an element $v \in \V^{\otimes L}$ has $\mathfrak{h}$-weight $\epsilon$ if $h  v = \epsilon  v$ for all $h \in \mathfrak{h}$. 
In addition to that we assign the weight $( \epsilon, \lambda_{\mathcal{A}}(x) , \lambda_{\mathcal{D}}(x) )$ to an element $v \in \V^{\otimes L}$
having $\mathfrak{h}$-weight $\epsilon$; and also satisfying $\mathcal{A} (x) v = \lambda_{\mathcal{A}}(x) v$, $\mathcal{D} (x) v = \lambda_{\mathcal{D}}(x) v$.
Here we restrict our attention to the case where $\lambda_{\mathcal{A}}$ and $\lambda_{\mathcal{D}}$ are non-vanishing meromorphic functions.
We then proceed by defining \emph{singular vectors} as the elements $v \in \V^{\otimes L}$ such that $\mathcal{C} (x) v = 0$ for generic $x \in \C$. 
Highest-weight modules are then formed by singular vectors $v \in \V^{\otimes L}$ exhibiting weight 
$( \epsilon, \lambda_{\mathcal{A}}(x) , \lambda_{\mathcal{D}}(x) )$. Hence, taking into account \eqref{mono} and \eqref{rmat}, one readily
finds that 
\[ \label{zero}
\ket{0} \coloneqq \left(\begin{matrix} 1 \\ 0 \end{matrix} \right)^{\otimes L}
\]
is a highest-weight vector. Bethe vectors $\ket{\Psi_n} \in \text{span}(\V^{\otimes L})$ for $1 \leq n \leq L$ are then defined as
\[ \label{bethe}
\ket{\Psi_n} \coloneqq \PROD{1}{i}{n} \mathcal{B} (x_i) \ket{0} \qquad x_i \in \C \; .
\]
It is worth remarking that \eqref{bethe} describes \emph{off-shell} Bethe vectors. The \emph{on-shell} case is obtained from a particular specialization
of \eqref{bethe}. More precisely, we ask the variables $x_i$ in \eqref{bethe} to solve Bethe ansatz equations.

Next we discuss the construction of scalar products of Bethe vectors. In order to define such scalar products one needs to build suitable duals
of \eqref{bethe}. This analysis has been performed in \cite{Korepin_1982} and here we present only the final result. Following \cite{Korepin_1982}, off-shell scalar
products are functions $\mathscr{S}_n \colon \C^{2 n} \to \C$ defined as
\[ \label{SP}
\mathscr{S}_n (x_1 , \dots , x_n | x_1^{\mathcal{B}} , \dots , x_n^{\mathcal{B}} ) \coloneqq  \bra{0} \iPROD{1}{i}{n} \mathcal{C}(x_i) \PROD{1}{j}{n} \mathcal{B}(x_j^{\mathcal{B}}) \ket{0} \; .
\]

Following our previous discussion, the specialization of \eqref{SP} to on-shell scalar products are obtained by means of Bethe ansatz equations. In the present case 
case they read
\[
\label{BA}
\frac{\phi_1}{\phi_2} \prod_{i=1}^L \frac{a(\lambda_k - \mu_i)}{b(\lambda_k - \mu_i)} \prod_{\substack{j=1 \\ j \neq k}}^n \frac{a(\lambda_j - \lambda_k)}{a(\lambda_k - \lambda_j)} = (-1)^{n-1} 
\]
for $\phi_1, \phi_2 \in \C^{\times}$. Now let $\Omega = \bigcup_{\lambda \in \Lambda(\Omega)} \Omega[\lambda]$ be the solutions manifold of Eq. \eqref{BA}.
The submanifold $\Omega[\lambda]$ of dimension $n$ describes a particular solution and $\Lambda(\Omega)$ classifies all solutions
of the Bethe ansatz equations \eqref{BA}. It is worth remarking that a complete and rigorous classification of such solutions is to date an open problem.
In this way, on-shell scalar products $\mathcal{S}_n \colon \C^{n} \to \C$ are defined as the specialization
\[
\label{onSP}
\mathcal{S}_n (x_1 , \dots , x_n) \coloneqq \left. \mathscr{S}_n (x_1 , \dots , x_n | x_1^{\mathcal{B}} , \dots , x_n^{\mathcal{B}} )   \right|_{x_i^{\mathcal{B}} \in \Omega[\lambda]} \; .
\]
Strictly speaking one should also add a label $\lambda$ to the LHS of \eqref{onSP} but we will refrain to do so for the sake of simplicity. 
In what follows we address the problem of evaluating \eqref{onSP}.

\subsection{Functional equation} \label{sec:FZ_twist}
In the work \cite{Galleas_SCP} we have presented functional equations determining the scalar product \eqref{SP}. Among the results of \cite{Galleas_SCP} we
have also shown that the specialization \eqref{onSP} satisfy a simpler functional relation exhibiting the same structure as the equation
satisfied by the partition function of the six-vertex model with domain-wall boundaries \cite{Galleas_2013, Galleas_proc}. 
This simplified equation has also been studied in \cite{Galleas_Lamers_2015} and it is given as follows.

\begin{defi}[Sets of variables]
Write $X \coloneqq \{ x_1 , x_2 , \dots , x_n \}$ and introduce the symbol 
\[
X_{\alpha_1 , \alpha_2 , \dots , \alpha_p}^{\beta_1 , \beta_2 , \dots , \beta_q} \coloneqq X \cup \{x_{\beta_1} , x_{\beta_2} , \dots , x_{\beta_q} \} \backslash \{x_{\alpha_1} , x_{\alpha_2} , \dots , x_{\alpha_p} \}
\]
for $p,q \in \Z_{\geq 0}$. 
\end{defi}

\begin{thm} \label{funEQ}
On-shell scalar products $\mathcal{S}_n$ satisfy the functional relation
\[
\label{FS}
\sum_{i=0}^n K_i \; \mathcal{S}_n (X_i^0) = 0 
\]
with coefficients reading
\<
\label{coeff}
K_0 &\coloneqq& \phi_1 \prod_{j=1}^L a(x_0 - \mu_j) \left[ \prod_{k=1}^n \frac{a(x_k - x_0)}{b(x_k - x_0)} -  \prod_{k=1}^n \frac{a(x_k^{\mathcal{B}} - x_0)}{b(x_k^{\mathcal{B}} - x_0)}  \right] \nonumber \\
&& + \; \phi_2 \prod_{j=1}^L b(x_0 - \mu_j) \left[ \prod_{k=1}^n \frac{a(x_0 - x_k)}{b(x_0 - x_k)} -  \prod_{k=1}^n \frac{a(x_0 - x_k^{\mathcal{B}})}{b(x_0 - x_k^{\mathcal{B}})}  \right] \nonumber \\[5mm]
K_{i} &\coloneqq& \phi_1 \frac{c(x_0 - x_i)}{b(x_0 - x_i)} \prod_{j=1}^L a(x_i - \mu_j) \prod_{\substack{k=1 \\ k \neq i}}^n \frac{a(x_k - x_i)}{b(x_k - x_i)} \nonumber \\
&& + \; \phi_2 \frac{c(x_i - x_0)}{b(x_i - x_0)} \prod_{j=1}^L b(x_i - \mu_j) \prod_{\substack{k=1 \\ k \neq i}}^n \frac{a(x_i - x_k)}{b(x_i - x_k)} \qquad \qquad i = 1, 2, \dots , n  \; .
\> 
\end{thm}

The proof of Theorem \ref{funEQ} can be found in \cite{Galleas_SCP, Galleas_Lamers_2015} and it follows from a linear combination of the equations derived in \cite{Galleas_SCP}.

\subsection{Determinantal formula} \label{sec:DET_twist}

In order to solve Eq. \eqref{FS} we shall employ the method recently put forward in \cite{Galleas_2016a, Galleas_2016b}.
For that we first need to inspect the behavior of \eqref{FS} under the action of the symmetric group. 

\begin{defi}
Let ${\rm Sym} (X^0)$ denote the symmetric group of degree $n+1$ on the set $X^0$ and let $\pi_{i,j} \in {\rm Sym} (X^0)$ 
be a $2$-cycle acting as the permutation of variables $x_i \leftrightarrow x_j$. Also, let $\text{Fun} (\C^{n+1})$ denote the space
of meromorphic functions on $\C^{n+1}$. The action of $ \pi_{i,j} \in {\rm Sym} (X^0)$ on $f \in \text{Fun} (\C^{n+1})$  is then given by
\[
( \pi_{i,j} f ) (\dots , x_i , \dots , x_j , \dots ) \coloneqq f(\dots , x_j , \dots , x_i , \dots ) \; .
\]
\end{defi}

One important feature of Eq. \eqref{FS} is that it depends on $n+1$ variables, although the scalar product \eqref{onSP} depends only on
$n$ variables. In particular, one can readily notice the distinguished role played by the variable $x_0$ by inspecting the action of $\pi_{i,j}$ 
on \eqref{FS}. More precisely, we find that the action of $\pi_{i,j}$ for $1 \leq i, j \leq n$ leaves our equation invariant, whereas
$\pi_{0,j}$ produces a new equation with same structure but modified coefficients. The fundamental idea of the method of \cite{Galleas_2016a, Galleas_2016b}
is to use this asymmetry in our favor.

\begin{lem} \label{l0n}
Eq. \eqref{FS} extends to the following system of linear equations,
\[ \label{FSl}
\sum_{i=0}^n K_i^{(l)} \; \mathcal{S}_n (X_i^0) = 0 \; , 
\]
with coefficients defined as
\< \label{coeffl}
K_i^{(l)} \coloneqq \begin{cases}
\pi_{0, l} K_l \qquad \qquad i = 0 \\
\pi_{0, l} K_0 \qquad \qquad i = l \\
\pi_{0, l} K_i \qquad \qquad \text{otherwise}
\end{cases} \; .
\>
\end{lem}
\begin{proof}
Straightforward application of $\pi_{0, l}$ on Eq. \eqref{FS} keeping in mind that $\mathcal{S}_n$ is a symmetric function.
\end{proof}

\begin{rema} \label{rm1}
The index $l$ in \eqref{FSl} runs from $0$ to $n$; hence, the system of equations stated in Lemma \ref{l0n} totals $n+1$ equations.
The case $l=0$ is regarded as our original equation. Also, each equation in \eqref{FSl} contains $n+1$ terms of the form $\mathcal{S}_n (X_i^0)$ 
and one can verify that $\text{det} \left( K_i^{(l)} \right)_{0 \leq i, l \leq n} = 0$. The latter property ensures the existence of non-trivial solutions.
\end{rema}

Interestingly, the system of equations \eqref{FSl} can be solved through rather elementary methods. This is one important aspect
of having $\mathcal{S}_n$ described by linear equations. In particular, we shall obtain single determinantal representations for
the scalar product $\mathcal{S}_n$ by using \emph{linear algebra} and \emph{separation of variables}.

\begin{thm} \label{T:twist}
The on-shell scalar product $\mathcal{S}_n$ can be written as
\< \label{SNdet}
\mathcal{S}_n (X) &=&  \phi_1^{-n} \; {\rm det} (\mathcal{V})  \prod_{j=1}^n \frac{b(x_0 - x_j)}{b(x_j^{\mathcal{B}} - x_0)} \prod_{k=1}^L b(x_j^{\mathcal{B}} - \mu_k) \nonumber \\
\mathcal{S}_n (X_i^0) &=& \phi_1^{-n} \; {\rm det} (\mathcal{V}_i)  \prod_{j=1}^n \frac{b(x_0 - x_j)}{b(x_j^{\mathcal{B}} - x_0)} \prod_{k=1}^L b(x_j^{\mathcal{B}} - \mu_k) \; , \nonumber \\
\>
where $\mathcal{V}$ and $\mathcal{V}_i$ are matrices of dimension $n \times n$. Their entries read
\< \label{entries}
\mathcal{V}_{\alpha , \beta} &=& \begin{cases}
\displaystyle \phi_1 \prod_{j=1}^L a(x_{\alpha} - \mu_j) \left[ \prod_{\substack{k=0 \\ k \neq \alpha}}^n \frac{a(x_k - x_{\alpha})}{b(x_k - x_{\alpha})} -  \prod_{k=1}^n \frac{a(x_k^{\mathcal{B}} - x_{\alpha})}{b(x_k^{\mathcal{B}} - x_{\alpha})}  \right] \nonumber \\
\displaystyle \qquad + \; \phi_2 \prod_{j=1}^L b(x_{\alpha} - \mu_j) \left[ \prod_{\substack{k=0 \\ k \neq \alpha}}^n \frac{a(x_{\alpha} - x_k)}{b(x_{\alpha} - x_k)} -  \prod_{k=1}^n \frac{a(x_{\alpha} - x_k^{\mathcal{B}})}{b(x_{\alpha} - x_k^{\mathcal{B}})}  \right] \qquad \beta = \alpha \nonumber \\[10mm]
\displaystyle \phi_1 \frac{c(x_{\alpha} - x_{\beta})}{b(x_{\alpha} - x_{\beta})} \prod_{j=1}^L a(x_{\beta} - \mu_j) \prod_{\substack{k=0 \\ k \neq \alpha , \beta}}^n \frac{a(x_k - x_{\beta})}{b(x_k - x_{\beta})} \nonumber \\
\displaystyle \qquad + \; \phi_2 \frac{c(x_{\beta} - x_{\alpha})}{b(x_{\beta} - x_{\alpha})} \prod_{j=1}^L b(x_{\beta} - \mu_j)  \prod_{\substack{k=0 \\ k \neq \alpha, \beta}}^n \frac{a(x_{\beta} - x_k)}{b(x_{\beta} - x_k)} \hfill \mbox{otherwise}
\end{cases} \\
\>
\<
(\mathcal{V}_i)_{\alpha , \beta} &=& \begin{cases}
\displaystyle - \phi_1 \frac{c(x_{\alpha} - x_{0})}{b(x_{\alpha} - x_{0})} \prod_{j=1}^L a(x_{0} - \mu_j) \prod_{\substack{k=1 \\ k \neq \alpha}}^n \frac{a(x_k - x_{0})}{b(x_k - x_{0})} \nonumber \\
\displaystyle \qquad - \; \phi_2 \frac{c(x_{0} - x_{\alpha})}{b(x_{0} - x_{\alpha})} \prod_{j=1}^L b(x_{0} - \mu_j)  \prod_{\substack{k=1 \\ k \neq \alpha}}^n \frac{a(x_{0} - x_k)}{b(x_{0} - x_k)} \hfill \beta = i \nonumber \\[10mm]
\displaystyle \phi_1 \prod_{j=1}^L a(x_{\alpha} - \mu_j) \left[ \prod_{\substack{k=0 \\ k \neq \alpha}}^n \frac{a(x_k - x_{\alpha})}{b(x_k - x_{\alpha})} -  \prod_{k=1}^n \frac{a(x_k^{\mathcal{B}} - x_{\alpha})}{b(x_k^{\mathcal{B}} - x_{\alpha})}  \right] \nonumber \\
\displaystyle \qquad + \; \phi_2 \prod_{j=1}^L b(x_{\alpha} - \mu_j) \left[ \prod_{\substack{k=0 \\ k \neq \alpha}}^n \frac{a(x_{\alpha} - x_k)}{b(x_{\alpha} - x_k)} -  \prod_{k=1}^n \frac{a(x_{\alpha} - x_k^{\mathcal{B}})}{b(x_{\alpha} - x_k^{\mathcal{B}})}  \right] \qquad \beta \neq i , \beta = \alpha \nonumber \\[10mm]
\displaystyle \phi_1 \frac{c(x_{\alpha} - x_{\beta})}{b(x_{\alpha} - x_{\beta})} \prod_{j=1}^L a(x_{\beta} - \mu_j) \prod_{\substack{k=0 \\ k \neq \alpha , \beta}}^n \frac{a(x_k - x_{\beta})}{b(x_k - x_{\beta})} \nonumber \\
\displaystyle \qquad + \; \phi_2 \frac{c(x_{\beta} - x_{\alpha})}{b(x_{\beta} - x_{\alpha})} \prod_{j=1}^L b(x_{\beta} - \mu_j)  \prod_{\substack{k=0 \\ k \neq \alpha, \beta}}^n \frac{a(x_{\beta} - x_k)}{b(x_{\beta} - x_k)} \hfill \beta \neq i , \alpha
\end{cases} \\
\>
\end{thm}
\begin{proof}
Among the $n+1$ equations contained in \eqref{FSl} we select the ones in the range $1 \leq l \leq n$. This totals $n$ equations, and 
as pointed out in Remark \ref{rm1}, the system \eqref{FSl} contains $n+1$ terms of the form $\mathcal{S}_n (X_i^0)$. 
Hence, we can use these $n$ equations $(1 \leq l \leq n)$ to express each component $\mathcal{S}_n (X_i^0)$ for $1 \leq i \leq n$ 
in terms of $\mathcal{S}_n (X)$. This procedure can be algorithmically implemented using Cramer's rule. By doing so we find
\[ \label{S2S}
\mathcal{S}_n (X_i^0) = \frac{{\rm det} (\mathcal{V}_i) }{{\rm det} (\mathcal{V})} \mathcal{S}_n (X) 
\]
with matrices $\mathcal{V}$ and $\mathcal{V}_i$ defined in terms of the coefficients $K_i^{(l)}$. More precisely, the entries
of such matrices are given by
\< \label{VVI}
\mathcal{V}_{\alpha , \beta} &\coloneqq& K_{\beta}^{(\alpha)} \qquad 1 \leq \alpha , \beta \leq n \\
(\mathcal{V}_i)_{\alpha , \beta} &\coloneqq& \begin{cases}
- K_{0}^{(\alpha)} \qquad \beta = i \\
K_{\beta}^{(\alpha)} \qquad \mbox{otherwise}
\end{cases} \; .
\>

Now, as discussed in \cite{Galleas_2016b}, Eq. \eqref{S2S} can be regarded as an one-variable functional equation, i.e. $\mathcal{S}_n (x_0) \sim \mathcal{S}_n (x_i)$,
and since \eqref{S2S} holds for all $i \in \{ 1, 2 , \dots , n \}$ we can conclude that
\[ \label{sov}
{\rm det} (\mathcal{V}) =  \mathcal{S}_n (X) f(x_0, x_1 , \dots , x_n) \qquad \mbox{and} \qquad  {\rm det} (\mathcal{V}_i) =  \mathcal{S}_n (X_i^0) f(x_0, x_1 , \dots , x_n) 
\]
for a given function $f$. The proof of \eqref{SNdet} then follows from the determination of the function $f$ in \eqref{sov}.

We proceed by noticing that the decompositions \eqref{sov} implies a system of partial differential equations. More precisely, we have
$\partial_{0} ( f / {\rm det}(\mathcal{V}) ) = 0$ and $\partial_{i} ( f / {\rm det}(\mathcal{V}_i) ) = 0$ where
$\partial_{\alpha} \coloneqq \partial / \partial x_{\alpha}$. Such equations can be easily solved and we find
\[ \label{f}
f(x_0 , x_1 , \dots , x_n) = \omega \prod_{j=1}^n \frac{b(x_0 - x_j^{\mathcal{B}})}{b(x_0 - x_j)} \; , 
\]
with $\omega$ being an integration constant. The latter can be fixed from the asymptotic behavior presented in \cite{Galleas_SCP}
and we obtain 
\[ \label{omega}
\omega =  \frac{(-1)^n \phi_1^n}{ \displaystyle \prod_{k=1}^n \prod_{l=1}^L b(x_k^{\mathcal{B}} - \mu_l)} \; .
\]
The substitution of \eqref{coeff}, \eqref{coeffl}, \eqref{VVI}, \eqref{f} and \eqref{omega} in \eqref{sov} concludes our proof.
\end{proof}

\begin{rema}
The first formula in \eqref{SNdet} can be absorbed by the second by extending the index $i$ to the value $i=0$; while keeping in mind the entries
$(\mathcal{V}_0)_{\alpha , \beta}$ are defined for $1 \leq \alpha, \beta \leq n$. However, we prefer to keep both formulae in \eqref{SNdet}
separated for the sake of clarity.
\end{rema}

Some comments regarding the determinantal representations in \eqref{SNdet} are in order. For instance, the first formula expresses the scalar
product $\mathcal{S}_n (X)$ in terms of the determinant of the matrix $\mathcal{V}$ whose entries depend explicitly on the variable $x_0$. This dependence is highly non-trivial and its explicit form can be seen from \eqref{entries}. Nevertheless,
the evaluation of the RHS leaves us with an expression depending only on $X$. In this way, the dependence on $x_0$ in the determinantal
representation for $\mathcal{S}_n (X)$ is \emph{local} but not \emph{global}, which allows one to fix $x_0$ at convenience. Hence,
we regard \eqref{SNdet} as a continuous determinantal representations parameterized by this extra variable. It is also worth remarking
that this feature is the same mechanism allowing for the evaluation of path integrals through the \emph{localization method} 
\cite{Duistermaat_Heckman_1982}. 

On the other hand, the second formula in \eqref{SNdet} expresses $\mathcal{S}_n (X_i^0)$ in terms of the determinant of the matrix $\mathcal{V}_i$. 
Both formulae can be compared through the identification $\pi_{0,i} \mathcal{S}_n (X_i^0) =  \mathcal{S}_n (X)$; however, we have not found a 
relation between the determinants ${\rm det} ( \pi_{0,i} \mathcal{V}_i )$ and ${\rm det} ( \mathcal{V} )$. This suggests formulae \eqref{SNdet} indeed consist of independent
representations. Moreover, under this assumption \eqref{SNdet}  totals $n+1$ families of continuous determinantal representations for the scalar
product $\mathcal{S}_n$.

\section{Six-vertex model: open boundaries} \label{sec:OPEN}

In \Secref{sec:TWIST} we have studied on-shell scalar products of Bethe vectors in the six-vertex model with toroidal boundary conditions
by means of functional equations. Here we intend to extend that study to the case with open boundary conditions.
The starting point of our analysis is the construction of Bethe vectors and, interestingly, the series of developments leading to the algebraic 
construction of Bethe vectors in the six-vertex model with open boundaries was analogous to the toroidal case. 
The history began with the generalization of the Bethe ansatz for quantum hamiltonians with open boundary conditions in
the works \cite{Gaudin_1971, Gaudin_book, Alcaraz_1987}. Subsequently, the scattering theory (on the half-line) was studied by
Cherednik \cite{Cherednik_1984} and the condition of factorized scattering led to the formulation of the so called 
\emph{reflection equations}. The latter was put forward as an analogous of the Yang-Baxter equation governing scattering
at the boundaries. The interplay between these results was then clarified with the formulation of the \emph{Boundary Quantum Inverse Scattering Method} 
by Sklyanin in \cite{Sklyanin_1988}. One of the outcomes of Sklyanin's formulation is the algebraic construction of Bethe vectors in the six-vertex
model with open boundary conditions. In what follows we shall briefly review such construction and define scalar products of interest.

\subsection{Definitions and conventions} \label{sec:DEF_open}

The formulation of the Boundary Quantum Inverse Scattering Method, as described by Sklyanin in \cite{Sklyanin_1988}, requires a few more ingredients
in addition to the ones already described in \Secref{sec:TWIST}. In this section we give a brief account of those extra ingredients and, for more details,
we refer the reader to the original work \cite{Sklyanin_1988}. Here the conventions of \cite{Galleas_openSCP} are also extensively used.

\subsection{Reflection matrices} Following \cite{Cherednik_1984, Sklyanin_1988} we introduce \emph{reflection matrices} characterizing the system's 
interactions at the boundaries. More precisely, let $\mathcal{K}, \bar{\mathcal{K}} \colon \C \to \text{End} (\V)$ be respectively the reflection matrix 
and its dual. Integrability in the sense of Baxter then requires such matrices to satisfy the \emph{reflection equations}. In contrast to the 
Yang-Baxter equation governing the $\mathcal{R}$-matrix \eqref{rmat}, which is a relation in $\text{End} (\V \otimes \V \otimes \V)$, 
the reflection equations are relations in $\text{End} (\V \otimes \V)$ reading
\< \label{k1}
&&\mathcal{R}_{12} (x_1 - x_2) \mathcal{K}_1 (x_1) \mathcal{R}_{12} (x_1 + x_2) \mathcal{K}_2 (x_2) \nonumber \\
&&\qquad = \mathcal{K}_2 (x_2) \mathcal{R}_{12} (x_1 + x_2) \mathcal{K}_1 (x_1) \mathcal{R}_{12} (x_1 - x_2) \\[5mm] \label{k2}
&&\mathcal{R}_{12} (- x_1 + x_2) \bar{\mathcal{K}}_1^{t_1} (x_1) \mathcal{R}_{12} (- x_1 - x_2 - 2\gamma) \bar{\mathcal{K}}_2^{t_2} (x_2) \nonumber \\
&&\qquad = \bar{\mathcal{K}}_2^{t_2} (x_2) \mathcal{R}_{12} (- x_1 - x_2 - 2\gamma) \bar{\mathcal{K}}_1^{t_1} (x_1) \mathcal{R}_{12} (- x_1 + x_2) \; . 
\>
In \eqref{k1} and \eqref{k2} we use the standard tensor leg notation and $t_i$ stands for the matrix transposition in the $i$-th
space of  $\V \otimes \V$. The matrix $\mathcal{R}$ enters as input in the reflection equations and a general solution of
\eqref{k1} and \eqref{k2} associated with \eqref{rmat} was obtained in \cite{Ghoshal_Zamolodchikov_1994}.
Here we shall consider a particular solution of \eqref{k1} and \eqref{k2}, namely
\<
\mathcal{K} (x) = \begin{pmatrix} \sinh{(h+x)} & 0 \\ 0 & \sinh{(h-x)} \end{pmatrix}  , \quad  \bar{\mathcal{K}} (x) = \begin{pmatrix} \sinh{(\bar{h} - x - \gamma)} & 0 \\ 0 & \sinh{(\bar{h} + x + \gamma)} \end{pmatrix}  \; , \nonumber \\
\>
with fixed parameters $h, \bar{h} \in \C$.

\subsection{Reflection algebra} It turns out one can also formulate an associative algebra, similar to the Yang-Baxter algebra,
suitable for describing systems with open boundaries conditions. Let $\mathscr{T} (\mathcal{R})$ denote the algebra generated by the 
non-commutative matrix entries of $\mathscr{L} \in \text{End} (\V)$ satisfying
\< \label{RA}
&&\mathcal{R}_{12} (x_1 - x_2) \mathscr{L}_1 (x_1) \mathcal{R}_{12} (x_1 + x_2) \mathscr{L}_2 (x_2) \nonumber \\
&&\qquad = \mathscr{L}_2 (x_2) \mathscr{L}_{12} (x_1 + x_2) \mathscr{L}_1 (x_1) \mathcal{R}_{12} (x_1 - x_2) \; .
\>
The relation \eqref{RA} lives in $\text{End} (\V \otimes \V \otimes \V_{\mathcal{Q}})$ and the matrix $\mathcal{R}$ plays the role of 
structure constant for $\mathscr{T} (\mathcal{R})$. In particular, here $\V \simeq \C^2$ and we write
\[ \label{abcdr}
\mathscr{L} (x) \eqqcolon \begin{pmatrix} \widetilde{\mathcal{A}}(x) & \widetilde{\mathcal{B}}(x) \\ \widetilde{\mathcal{C}}(x) & \widetilde{\mathcal{D}}(x) \end{pmatrix}
\]
with operator-valued elements $\widetilde{\mathcal{A}}, \widetilde{\mathcal{B}}, \widetilde{\mathcal{C}}, \widetilde{\mathcal{D}} \colon \C \to \text{End} (\V_{\mathcal{Q}})$. We refer to $\mathscr{T} (\mathcal{R})$ as \emph{reflection algebra}.

\subsection{Modules over $\mathscr{T} (\mathcal{R})$} Let the pair $(\V_{\mathcal{Q}} , \mathscr{L})$ be an $\mathscr{T} (\mathcal{R})$-module. In particular,
due to the reflection equation \eqref{k1} one can readily notice that for $\V_{\mathcal{Q}} = \C$ we can take $\mathscr{L} (x) = \mathcal{K} (x)$.
However, in order to construct Bethe vectors we shall need representations of the reflection algebra, or $\mathscr{T} (\mathcal{R})$-modules, with
$\V_{\mathcal{Q}} = \V^{\otimes L}$. This is given by the following theorem.

\begin{thm}[Sklyanin] The pair $(\V^{\otimes L}, \widetilde{\mathcal{T}}_0 )$ is a $\mathscr{T} (\mathcal{R})$-module with
\[ \label{dmono}
\widetilde{\mathcal{T}}_0 (x) \coloneqq \bar{\mathcal{K}}_0 (x) \iPROD{1}{j}{L} \mathcal{R}_{0 j} (x - \mu_j) \; \mathcal{K}_0 (x) \PROD{1}{j}{L} \mathcal{R}_{0 j} (x + \mu_j) \;\; \in {\rm End} (\V_0 \otimes \V^{\otimes L})  \; .
\]
\end{thm}
Therefore, representations of the operators $\widetilde{\mathcal{A}}, \widetilde{\mathcal{B}}, \widetilde{\mathcal{C}}, \widetilde{\mathcal{D}} \in \text{End}(\V^{\otimes})$ are obtained through the identification of \eqref{abcdr}
and \eqref{dmono}.

\subsection{Bethe vectors and scalar products} As far as $\widetilde{\mathcal{T}}_0$ built out of \eqref{rmat}, \eqref{k1} and \eqref{k2} is concerned, the
construction of Bethe vectors in the six-vertex model with open boundaries follows the same lines presented in \Secref{sec:TWIST}. 
Let $\ket{0}$ be the highest-weight vector defined in \eqref{zero}. Then following \cite{Sklyanin_1988}, off-shell Bethe vectors $\ket{\Phi_n} \in \text{span}(\V^{\otimes L})$ 
are defined as
\[ \label{bethe_open}
\ket{\Phi_n} \coloneqq \PROD{1}{i}{n} \widetilde{\mathcal{B}} (x_i) \ket{0} \qquad x_i \in \C \; .
\]
In the present paper we are interested in on-shell scalar products and that is obtained by settling \eqref{bethe_open} on the manifold generated by
appropriated Bethe ansatz equations. In our case such equations read
\< \label{BAopen}
\frac{b(\lambda_k + h) b(\lambda_k - \bar{h})}{a(\lambda_k - h) a(\lambda_k + \bar{h})} \prod_{i=1}^L \frac{a(\lambda_k - \mu_i) a(\lambda_k + \mu_i)  }{b(\lambda_k - \mu_i) b(\lambda_k + \mu_i) } \prod_{\substack{j=1 \\ j \neq k}}^n \frac{a(\lambda_j - \lambda_k) b(\lambda_j + \lambda_k)}{a(\lambda_k - \lambda_j) a(\lambda_k + \lambda_j + \gamma)} = (-1)^{n-1} \; . \nonumber \\ 
\>
Now let $\Gamma = \bigcup_{\lambda \in \Upsilon(\Gamma)} \Gamma[\lambda]$ be the manifold constituted by submanifolds $\Gamma[\lambda]$
describing particular solutions of \eqref{BAopen}. In its turn, $\Upsilon(\Gamma)$ characterizes all sets of solutions. On-shell Bethe vectors
are then the specialization $\left. \ket{\Phi_n} \right|_{x_i \in \Gamma[\lambda]}$.

We stress again that here we are considering \eqref{k1} and \eqref{k2}; and in that case the formulation of scalar products of Bethe vectors 
is not significantly different from the one presented in \Secref{sec:TWIST}. Essentially, one only needs to consider
\eqref{abcdr} instead of \eqref{abcd}. Hence, we write 
\[ \label{openSP}
\widetilde{\mathscr{S}}_n (x_1 , \dots , x_n | x_1^{\widetilde{\mathcal{B}}} , \dots , x_n^{\widetilde{\mathcal{B}}} ) \coloneqq  \bra{0} \iPROD{1}{i}{n} \widetilde{\mathcal{C}}(x_i) \PROD{1}{j}{n} \widetilde{\mathcal{B}}(x_j^{\widetilde{\mathcal{B}}}) \ket{0} \; 
\]
for off-shell scalar products. Definition \eqref{openSP} describes a multivariate function $\widetilde{\mathscr{S}}_n \colon \C^{2n} \to \C$
and, similarly to the case with toroidal boundaries, we define on-shell scalar products as a particular specialization of \eqref{openSP}. More precisely,
on-shell scalar products in the six-vertex model with open boundaries are defined as
\[
\label{onopenSP}
\widetilde{\mathcal{S}}_n (x_1 , \dots , x_n) \coloneqq \left. \widetilde{\mathscr{S}}_n (x_1 , \dots , x_n | x_1^{\widetilde{\mathcal{B}}} , \dots , x_n^{\widetilde{\mathcal{B}}} )   \right|_{x_i^{\widetilde{\mathcal{B}}} \in \Gamma[\lambda]} \; .
\]
Hence, $\widetilde{\mathcal{S}}_n \colon \C^n \to \C$ and we shall discuss the evaluation of $\widetilde{\mathcal{S}}_n$ in the following subsection.

\subsection{Functional equation} \label{sec:FZ_open}
In the works \cite{Galleas_openSCP, Galleas_Lamers_2015} we have discussed a functional equation solved by the on-shell scalar product \eqref{onopenSP}.
This equation originates from the reflection algebra \eqref{RA} and here we intend to use it to obtain explicit representations for $\widetilde{\mathcal{S}}_n$.
The derivation of such equation within the algebraic-functional framework can be found in \cite{Galleas_openSCP} and here we restrict ourselves
to presenting only the final result.

\begin{thm} \label{funEQ_open}
The scalar product $\widetilde{\mathcal{S}}_n$ satisfies the functional relation
\[
\label{FS_open}
\sum_{i=0}^n \widetilde{K}_i \; \widetilde{\mathcal{S}}_n (X_i^0) = 0 
\]
with coefficients defined as
\< \label{coeff_open}
\widetilde{K}_0 &\coloneqq&  b(x_0+h)b(x_0-\bar h) \frac{a(2 x_0+\gamma)}{b(2 x_0+\gamma)} \prod_{j=1}^L a(x_0 - \mu_j)a(x_0 + \mu_j) \nonumber \\ 
&& \qquad \times \left[ \prod_{k=1}^n \frac{a(x_k - x_0)}{b(x_k - x_0)}\frac{b(x_k + x_0)}{a(x_k + x_0)} -  \prod_{k=1}^n \frac{a(x_k^{\widetilde{\mathcal{B}}} - x_0)}{b(x_k^{\widetilde{\mathcal{B}}} - x_0)}\frac{b(x_k^{\widetilde{\mathcal{B}}} + x_0)}{a(x_k^{\widetilde{\mathcal{B}}} + x_0)} \right]  \nonumber \\
&& \;\;\; + \; a(x_0-h)a(x_0+\bar{h}) \frac{b(2 x_0)}{a(2 x_0)} \prod_{j=1}^L b(x_0 - \mu_j)b(x_0 + \mu_j) \nonumber \\ 
&& \qquad \times \left[ \prod_{k=1}^n \frac{a(x_0 - x_k)}{b(x_0 - x_k)}\frac{a(x_0 + x_k + \gamma)}{b(x_0 + x_k + \gamma)} -  \prod_{k=1}^n \frac{a(x_0 - x_k^{\widetilde{\mathcal{B}}})}{b(x_0 - x_k^{\widetilde{\mathcal{B}}})}\frac{a(x_0 + x_k^{\widetilde{\mathcal{B}}} + \gamma)}{b(x_0 + x_k^{\widetilde{\mathcal{B}}} + \gamma)}  \right] \nonumber \\[5mm]
\widetilde{K}_{i} &\coloneqq& \frac{a(2 x_0 +\gamma)}{a(x_0 + x_i)}\frac{c(x_0 - x_i)}{b(x_0 - x_i)} \frac{b(2 x_i)}{a(2 x_i)} \nonumber \\ 
&& \times \left[ b(x_i + h)b(x_i - \bar{h}) \prod_{j=1}^L a(x_i - \mu_j) a(x_i + \mu_j) \prod_{\substack{k=1 \\ k \neq i}}^n \frac{a(x_k - x_i)}{b(x_k - x_i)}\frac{b(x_k + x_i)}{a(x_k + x_i)} \right. \nonumber \\
&& \qquad \left. - \; a(x_i - h) a(x_i + \bar{h}) \prod_{j=1}^L b(x_i - \mu_j) b(x_i + \mu_j) \prod_{\substack{k=1 \\ k \neq i}}^n \frac{a(x_i - x_k)}{b(x_i - x_k)}\frac{a(x_i + x_k + \gamma)}{b(x_i + x_k + \gamma)} \right]  \nonumber \\
&& \qquad\qquad\qquad\qquad\qquad\qquad\qquad\qquad\qquad\qquad\qquad\qquad\qquad\qquad  i = 1, 2, \dots , n \; .	
\> 
\end{thm}
The proof of Theorem \ref{funEQ_open} can be found in \cite{Galleas_Lamers_2015}. Similarly to the toroidal case whose functional equation
for on-shell scalar products is given in Theorem \ref{funEQ}, Eq. \eqref{FS_open} also follows from a particular linear combination
of functional relations satisfied by off-shell scalar products. Next we shall discuss the resolution
of \eqref{FS_open} in terms of determinants.

\subsection{Determinantal formula} \label{sec:DET_open}
One particularly important feature of the algebraic-functional approach is that this method yields the same type of functional relation
for several quantities in integrable vertex models. For instance, one can readily notice that \eqref{FS} and \eqref{FS_open} only differ
by the explicit form of their coefficients. In addition to that the method recently put forward in \cite{Galleas_2016a, Galleas_2016b} 
relies mostly on the structure of the equation rather than on its particular coefficients. In this way, the solution of \eqref{FS_open} 
can be obtained straightforwardly and we present it in what follows.

\begin{lem} The action of ${\rm Sym}(X^0)$ on \eqref{FS_open} yields the system of equations
\[ \label{FSl_open}
\sum_{i=0}^n \widetilde{K}_i^{(l)} \; \widetilde{\mathcal{S}}_n (X_i^0) = 0 
\]
with coefficients defined as
\< \label{coeffl}
\widetilde{K}_i^{(l)} \coloneqq \begin{cases}
\pi_{0, l} \widetilde{K}_l \qquad \qquad i = 0 \\
\pi_{0, l} \widetilde{K}_0 \qquad \qquad i = l \\
\pi_{0, l} \widetilde{K}_i \qquad \qquad \text{otherwise}
\end{cases} \; .
\>
\end{lem}
\begin{proof}
First notice that \eqref{FS_open} is invariant under the action of $\pi_{i,j}$ for $1 \leq i,j \leq n$. The remaining elements
of ${\rm Sym}(X^0)$ are then $\pi_{0 , l}$ $(1 \leq l \leq n)$ whose action then produces \eqref{FSl_open}. The latter
statement also uses the fact that $\widetilde{\mathcal{S}}_n$ is a symmetric function.
\end{proof}

\begin{rema}
Expression \eqref{FSl_open} represents a system of $n+1$ equations, similarly to equations \eqref{FSl} described in \Secref{sec:DET_twist}.
Also, each equation in \eqref{FSl_open} relates $n+1$ terms of the form $\mathcal{S}_n (X_i^0)$;
and the existence of non-trivial solutions is ensured by the property $\text{det} \left( \widetilde{K}_i^{(l)} \right)_{0 \leq i, l \leq n} = 0$.
\end{rema}

The extension of \eqref{FS_open} to \eqref{FSl_open} allows one to use \emph{linear algebra} methods to obtain an explicit
formula for on-shell scalar products $\widetilde{\mathcal{S}}_n$. In fact, we find several families of determinantal representations
which are stated in the following theorem.

\begin{thm} \label{T:open}
The scalar product $\widetilde{\mathcal{S}}_n$ admits the following representations:
\< \label{SNdeta_open}
\widetilde{\mathcal{S}}_n (X) &=& {\rm det} (\widetilde{\mathcal{V}}) \prod_{1 \leq j < k \leq n} \frac{a(x_j^{\widetilde{\mathcal{B}}} + x_k^{\widetilde{\mathcal{B}}} + \gamma)}{b(x_j^{\widetilde{\mathcal{B}}} + x_k^{\widetilde{\mathcal{B}}})} \nonumber \\
&& \times \prod_{j=1}^n \frac{b(x_0 - x_j) a(x_0 + x_j)}{b(x_0 - x_j^{\widetilde{\mathcal{B}}}) a(x_0 + x_j^{\widetilde{\mathcal{B}}})} \frac{b(2 x_j^{\widetilde{\mathcal{B}}})}{a(2 x_j + \gamma)} \frac{a(x_j^{\widetilde{\mathcal{B}}} - h)}{b(x_j^{\widetilde{\mathcal{B}}} - \bar{h})} \prod_{k=1}^L b(x_j^{\widetilde{\mathcal{B}}} - \mu_k) b(x_j^{\widetilde{\mathcal{B}}} + \mu_k) \; , \nonumber \\
\>
\< \label{SNdetb_open}
\widetilde{\mathcal{S}}_n (X_i^0) &=& {\rm det} (\widetilde{\mathcal{V}}_i) \prod_{1 \leq j < k \leq n} \frac{a(x_j^{\widetilde{\mathcal{B}}} + x_k^{\widetilde{\mathcal{B}}} + \gamma)}{b(x_j^{\widetilde{\mathcal{B}}} + x_k^{\widetilde{\mathcal{B}}})} \nonumber \\
&& \times \prod_{j=1}^n \frac{b(x_0 - x_j) a(x_0 + x_j)}{b(x_0 - x_j^{\widetilde{\mathcal{B}}}) a(x_0 + x_j^{\widetilde{\mathcal{B}}})} \frac{b(2 x_j^{\widetilde{\mathcal{B}}})}{a(2 x_j + \gamma)} \frac{a(x_j^{\widetilde{\mathcal{B}}} - h)}{b(x_j^{\widetilde{\mathcal{B}}} - \bar{h})} \prod_{k=1}^L b(x_j^{\widetilde{\mathcal{B}}} - \mu_k) b(x_j^{\widetilde{\mathcal{B}}} + \mu_k) \; . \nonumber \\
\>
The matrix coefficients of $\widetilde{\mathcal{V}}$ and $\widetilde{\mathcal{V}}_i$ are in their turn defined as
\< \label{entries}
&& \widetilde{\mathcal{V}}_{\alpha , \beta} = \nonumber \\
&& \begin{cases}
\displaystyle b(x_{\alpha} + h)b(x_{\alpha}-\bar h) \frac{a(2 x_{\alpha}+\gamma)}{b(2 x_{\alpha}+\gamma)} \prod_{j=1}^L a(x_{\alpha} - \mu_j)a(x_{\alpha} + \mu_j) \nonumber \\ 
\displaystyle \qquad \times \left[ \prod_{\substack{k=0 \\ k \neq \alpha}}^n \frac{a(x_k - x_{\alpha})}{b(x_k - x_{\alpha})}\frac{b(x_k + x_{\alpha})}{a(x_k + x_{\alpha})} -  \prod_{k=1}^n \frac{a(x_k^{\widetilde{\mathcal{B}}} - x_{\alpha})}{b(x_k^{\widetilde{\mathcal{B}}} - x_{\alpha})}\frac{b(x_k^{\widetilde{\mathcal{B}}} + x_{\alpha})}{a(x_k^{\widetilde{\mathcal{B}}} + x_{\alpha})} \right]  \nonumber \\
\displaystyle \;\;\; + \; a(x_{\alpha}-h)a(x_{\alpha}+\bar{h}) \frac{b(2 x_{\alpha})}{a(2 x_{\alpha})} \prod_{j=1}^L b(x_{\alpha} - \mu_j)b(x_{\alpha} + \mu_j) \nonumber \\ 
\displaystyle \qquad \qquad \times \left[ \prod_{\substack{k=0 \\ k \neq \alpha}}^n \frac{a(x_{\alpha} - x_k)}{b(x_{\alpha} - x_k)}\frac{a(x_{\alpha} + x_k + \gamma)}{b(x_{\alpha} + x_k + \gamma)} -  \prod_{k=1}^n \frac{a(x_{\alpha} - x_k^{\widetilde{\mathcal{B}}})}{b(x_{\alpha} - x_k^{\widetilde{\mathcal{B}}})}\frac{a(x_{\alpha} + x_k^{\widetilde{\mathcal{B}}} + \gamma)}{b(x_{\alpha} + x_k^{\widetilde{\mathcal{B}}} + \gamma)}  \right] \nonumber \\
\qquad\qquad\qquad\qquad\qquad\qquad\qquad\qquad\qquad\qquad\qquad\qquad\qquad\qquad\qquad \qquad \qquad    \beta = \alpha \nonumber \\[5mm]
\displaystyle \frac{a(2 x_{\alpha} +\gamma)}{a(x_{\alpha} + x_{\beta})}\frac{c(x_{\alpha} - x_{\beta})}{b(x_{\alpha} - x_{\beta})} \frac{b(2 x_{\beta})}{a(2 x_{\beta})} \nonumber \\ 
\displaystyle \times \left[ b(x_{\beta} + h)b(x_{\beta} - \bar{h}) \prod_{j=1}^L a(x_{\beta} - \mu_j) a(x_{\beta} + \mu_j) \prod_{\substack{k=0 \\ k \neq \alpha , \beta}}^n \frac{a(x_k - x_{\beta})}{b(x_k - x_{\beta})}\frac{b(x_k + x_{\beta})}{a(x_k + x_{\beta})} \right. \nonumber \\
\displaystyle \quad \left. - \; a(x_{\beta} - h) a(x_{\beta} + \bar{h}) \prod_{j=1}^L b(x_{\beta} - \mu_j) b(x_{\beta} + \mu_j) \prod_{\substack{k=0 \\ k \neq \alpha, \beta}}^n \frac{a(x_{\beta} - x_k)}{b(x_{\beta} - x_k)}\frac{a(x_{\beta} + x_k + \gamma)}{b(x_{\beta} + x_k + \gamma)} \right]  \nonumber \\
\hfill \mbox{otherwise}
\end{cases} \\
\>
\<
&& (\widetilde{\mathcal{V}}_i)_{\alpha , \beta} = \nonumber \\
&&\begin{cases}
\displaystyle \frac{a(2 x_{\alpha} +\gamma)}{a(x_{\alpha} + x_0)}\frac{c(x_{\alpha} - x_0)}{b(x_{\alpha} - x_0)} \frac{b(2 x_0)}{a(2 x_0)} \nonumber \\ 
\displaystyle \times \left[ b(x_0 + h)b(x_0 - \bar{h}) \prod_{j=1}^L a(x_0 - \mu_j) a(x_0 + \mu_j) \prod_{\substack{k=1 \\ k \neq \alpha}}^n \frac{a(x_k - x_0)}{b(x_k - x_0)}\frac{b(x_k + x_0)}{a(x_k + x_0)} \right. \nonumber \\
\displaystyle \quad \left. - \; a(x_0 - h) a(x_0 + \bar{h}) \prod_{j=1}^L b(x_0 - \mu_j) b(x_0 + \mu_j) \prod_{\substack{k=1 \\ k \neq \alpha}}^n \frac{a(x_0 - x_k)}{b(x_0 - x_k)}\frac{a(x_0 + x_k + \gamma)}{b(x_0 + x_k + \gamma)} \right]  \nonumber \\
\; \quad\qquad\qquad\qquad\qquad\qquad\qquad\qquad\qquad\qquad\qquad\qquad\qquad\qquad\qquad\qquad\qquad  \beta = i \nonumber \\[5mm]
\displaystyle b(x_{\alpha} + h)b(x_{\alpha}-\bar h) \frac{a(2 x_{\alpha}+\gamma)}{b(2 x_{\alpha}+\gamma)} \prod_{j=1}^L a(x_{\alpha} - \mu_j)a(x_{\alpha} + \mu_j) \nonumber \\ 
\displaystyle \qquad \times \left[ \prod_{\substack{k=0 \\ k \neq \alpha}}^n \frac{a(x_k - x_{\alpha})}{b(x_k - x_{\alpha})}\frac{b(x_k + x_{\alpha})}{a(x_k + x_{\alpha})} -  \prod_{k=1}^n \frac{a(x_k^{\widetilde{\mathcal{B}}} - x_{\alpha})}{b(x_k^{\widetilde{\mathcal{B}}} - x_{\alpha})}\frac{b(x_k^{\widetilde{\mathcal{B}}} + x_{\alpha})}{a(x_k^{\widetilde{\mathcal{B}}} + x_{\alpha})} \right]  \nonumber \\
\displaystyle \;\;\; + \; a(x_{\alpha}-h)a(x_{\alpha}+\bar{h}) \frac{b(2 x_{\alpha})}{a(2 x_{\alpha})} \prod_{j=1}^L b(x_{\alpha} - \mu_j)b(x_{\alpha} + \mu_j) \nonumber \\ 
\displaystyle \qquad \qquad \times \left[ \prod_{\substack{k=0 \\ k \neq \alpha}}^n \frac{a(x_{\alpha} - x_k)}{b(x_{\alpha} - x_k)}\frac{a(x_{\alpha} + x_k + \gamma)}{b(x_{\alpha} + x_k + \gamma)} -  \prod_{k=1}^n \frac{a(x_{\alpha} - x_k^{\widetilde{\mathcal{B}}})}{b(x_{\alpha} - x_k^{\widetilde{\mathcal{B}}})}\frac{a(x_{\alpha} + x_k^{\widetilde{\mathcal{B}}} + \gamma)}{b(x_{\alpha} + x_k^{\widetilde{\mathcal{B}}} + \gamma)}  \right] \nonumber \\
\qquad\qquad\qquad\qquad\qquad\qquad\qquad\qquad\qquad\qquad\qquad\qquad\qquad\qquad\qquad   \beta \neq i , \beta = \alpha \nonumber \\[5mm]
\displaystyle \frac{a(2 x_{\alpha} +\gamma)}{a(x_{\alpha} + x_{\beta})}\frac{c(x_{\alpha} - x_{\beta})}{b(x_{\alpha} - x_{\beta})} \frac{b(2 x_{\beta})}{a(2 x_{\beta})} \nonumber \\ 
\displaystyle \times \left[ b(x_{\beta} + h)b(x_{\beta} - \bar{h}) \prod_{j=1}^L a(x_{\beta} - \mu_j) a(x_{\beta} + \mu_j) \prod_{\substack{k=0 \\ k \neq \alpha , \beta}}^n \frac{a(x_k - x_{\beta})}{b(x_k - x_{\beta})}\frac{b(x_k + x_{\beta})}{a(x_k + x_{\beta})} \right. \nonumber \\
\displaystyle \quad \left. - \; a(x_{\beta} - h) a(x_{\beta} + \bar{h}) \prod_{j=1}^L b(x_{\beta} - \mu_j) b(x_{\beta} + \mu_j) \prod_{\substack{k=0 \\ k \neq \alpha, \beta}}^n \frac{a(x_{\beta} - x_k)}{b(x_{\beta} - x_k)}\frac{a(x_{\beta} + x_k + \gamma)}{b(x_{\beta} + x_k + \gamma)} \right]  \nonumber \\
\qquad\qquad\qquad\qquad\qquad\qquad\qquad\qquad\qquad\qquad\qquad\qquad\qquad\qquad\qquad\qquad  \beta \neq i , \alpha
\end{cases}
\>
\end{thm}

\begin{proof}
The proof is analogous to the proof of Theorem \ref{T:twist}. As for the last step, we use the asymptotic behavior 
derived in \cite{Galleas_openSCP}.
\end{proof}

\begin{rema}
Similarly to the representations \eqref{SNdet} obtained for the toroidal six-vertex model, formulae \eqref{SNdeta_open} and \eqref{SNdetb_open} also depends locally
on one extra variable, i.e. $x_0$, which has no global influence and can be fixed at convenience. In fact, all remarks made for \eqref{SNdet}
are also valid for \eqref{SNdeta_open} and \eqref{SNdetb_open}.
\end{rema}

\section{Concluding remarks} \label{sec:CONCL}

The main results of this paper are Theorems \ref{T:twist} and \ref{T:open}. They provide continuous families of single determinantal
representations for on-shell scalar products in two variants of the six-vertex model. As far as the toroidal six-vertex model is concerned,
such scalar products have been originally computed in \cite{Slavnov_1989} where a determinantal representation was found. However, although our
work also yields determinantal representations, there are crucial differences between our results and the one of \cite{Slavnov_1989}. 
The most significant difference is that each one of our representations forms a continuous family parameterized by an additional complex variable
having no global influence. Nevertheless, expressions \eqref{entries} exhibit a highly non-trivial dependence on this extra variable and
particular specializations can modify drastically the matrix we are taking the determinant. Furthermore, Theorem \ref{T:twist} not only contains a 
single formula but a total of $n+1$ (complex) continuous representations.
On-shell scalar products of Bethe vectors in the six-vertex model with open boundaries were previously studied in \cite{Wang_2002} and
\cite{Kitanine_2007}. The comparison between the results presented in Theorem \ref{T:open} and the one of \cite{Kitanine_2007} follows the same
lines already discussed for the toroidal case; and the differences previously pointed out also hold in the case of open boundaries.

One important aspect of our analysis is the origin of such representations. They naturally appear as the solution of certain functional equations
originating from the Yang-Baxter and reflection algebras. In particular, the existence of such continuous representations are intimately associated
with the structure of the aforementioned functional equations. It is important to remark here that the relation between continuous determinantal representations
and functional equations derived within the algebraic-functional framework has been only recently demonstrated in \cite{Galleas_2016a, Galleas_2016b}.
Moreover, our approach is \emph{constructive} in the sense that one does not need to make \emph{educated guesses} in order to obtain such 
determinantal formulae. 

As far as the resolution of Eqs. \eqref{FS} and \eqref{FS_open} is concerned, some remarks are in order. For instance, our method is solely based
on linear algebra and separation of variables; and uniqueness of the solution (up to an overall multiplicative factor) follows automatically 
from the method we are employing. Moreover, within our approach we do not make any assumption concerning the nature of the solution which suggests more general six-vertex models
can be accommodated in Theorems \ref{T:twist} and \ref{T:open}. Also, the structure of the matrices $\mathcal{V}$, $\mathcal{V}_i$, 
$\widetilde{\mathcal{V}}$ and $\widetilde{\mathcal{V}}_i$ do not depend on the particular functional dependence of the coefficients
of Eqs. \eqref{FS} and \eqref{FS_open}. In fact, we have only used the explicit form of such coefficients in the determination of the 
overall multiplicative factors accompanying the determinants.
 
Our present analysis and the results presented in \cite{Galleas_2016a, Galleas_2016b} suggests that the existence of determinantal representations
is intimately associated with functional relations of type \eqref{FS} and \eqref{FS_open}. However, we have also shown in \cite{Galleas_2016b} 
for the elliptic $\mathfrak{gl}_2$ elliptic solid-on-solid model that variations of \eqref{FS} and \eqref{FS_open} also admit determinantal
solutions. As previously remarked, several other quantities to date are known to be described by functional relations of type \eqref{FS} 
and variations. Such equations have been derived through the algebraic-functional method and, for instance, we find that off-shell scalar products
of Bethe vectors are governed by a particular system of functional equations as shown in \cite{Galleas_SCP, Galleas_openSCP}. The latter equations
can be regarded as variations of \eqref{FS} and \eqref{FS_open} and one might wonder if such off-shell equations can also be solved by determinants. 
It is worth remarking that determinantal representations for off-shell scalar products of Bethe vectors are not known to date. Although
the computation of correlation functions in the Heisenberg spin-chain along the lines of \cite{Korepin_book} requires only on-shell scalar products,
the off-shell case can still be regarded as the partition function of a six-vertex model with special boundary conditions similar to the case with
domain-wall boundaries. In this way, it would be desirable having off-shell scalar products also expressed as a single determinant and the possibility
of deriving them from the equations presented in \cite{Galleas_SCP, Galleas_openSCP} deserves further attention.

\bibliographystyle{alpha}
\bibliography{references}

\end{document}